\documentclass[11pt]{article}%
\usepackage[letterpaper,top=1cm, bottom=2cm, left=1.9cm, right=1.9cm]{geometry}
\usepackage{latexsym, amsfonts, amsthm,amssymb,amscd,amsmath}
\usepackage{enumerate}
\usepackage{exscale}
\usepackage{color} 
\usepackage{graphicx}
\usepackage{overpic} 
\usepackage{bm}
\usepackage{bbm}
\usepackage{fancyhdr}
\usepackage{subcaption}
\usepackage[normalem]{ulem}

\providecommand{\U}[1]{\protect\rule{.1in}{.1in}}
\newtheorem{theorem}{Theorem}[section]

\newtheorem{proposition}[theorem]{Proposition}

\theoremstyle{definition}

\newtheorem{remark}[theorem]{Remark}
\numberwithin{equation}{section}

\def\eps{\varepsilon}

\begin{document}

\title{A central bank strategy for defending  a currency peg}
\author{Eyal Neuman\thanks{Department of Mathematics, Imperial College London, SW7
2AZ London, United Kingdom. E-mail: eyaln13@gmail.com }\setcounter{footnote}{6}
\and Alexander Schied\thanks{Department of Statistics and Actuarial Science,
University of Waterloo,   ON, N2L 3G1, Canada. E-mail: alex.schied@gmail.com }
\and Chengguo Weng\thanks{Department of Statistics and Actuarial Science,
University of Waterloo, ON, N2L 3G1, Canada. E-mail: c2weng@uwaterloo.ca}
\and Xiaole Xue\thanks{School of Management, 
Shandong University, Jinan 250100,  China, and  Department of Statistics and Actuarial Science,
University of Waterloo,   ON, N2L 3G1, Canada. E-mail: xlxue@sdu.edu.cn}
\date{\normalsize First version: November 29, 2019\\ \normalsize This version: August 2, 2020}}
\maketitle

\begin{abstract}
We consider a central bank strategy for maintaining a two-sided currency target zone, in which an exchange rate of two currencies is forced to stay between two thresholds. To keep the exchange rate from breaking the prescribed barriers, the central bank is generating permanent price impact and thereby accumulating inventory in the foreign currency. Historical examples of failed target zones illustrate that this inventory can become problematic, in particular when there is an adverse macroeconomic trend in the market. We model this situation through  a continuous-time market impact model of Almgren--Chriss-type with    drift, in which the exchange rate is a diffusion process controlled by the price impact of the central bank's intervention strategy. The objective of the central bank is to enforce the target zone through a strategy that minimizes the accumulated inventory. We formulate this objective as a stochastic control problem with random time horizon. It is solved by reduction to a singular boundary value problem that was solved by Lasry and Lions (1989). Finally, we provide numerical simulations of optimally controlled exchange rate processes and the corresponding evolution of the central bank inventory.
\end{abstract}

\noindent{\textbf{Key words}.} Currency target zone, currency peg, price impact, central bank intervention, singular stochastic control, second-order differential equation with infinite boundary conditions

\noindent\textbf{MSC subject classifications.}  93E20, 49A60, 91G80

\section{Introduction}

Currency pegs, also called currency target zones, describe a regime in which a central bank imposes upper and/or lower limits on the exchange rate of the domestic currency against a foreign currency. There is an abundance of current and historical examples for such currency pegs, and in many cases, the termination of a currency peg had dramatic  consequences. A particularly notorious example is the exit of the British pound from the European Exchange Rate Mechanism (ERM) on September 16, 1992, a day known as \lq\lq Black Wednesday".

The goal of this note is to describe mathematically the situation that a central bank is facing when defending a currency peg against a macroeconomic trend. More precisely, we consider a  two-sided currency target zone, in which an exchange rate of two currencies is forced to stay between two thresholds. Two current examples of such target zones are shown in Figure \ref{figure1}. 
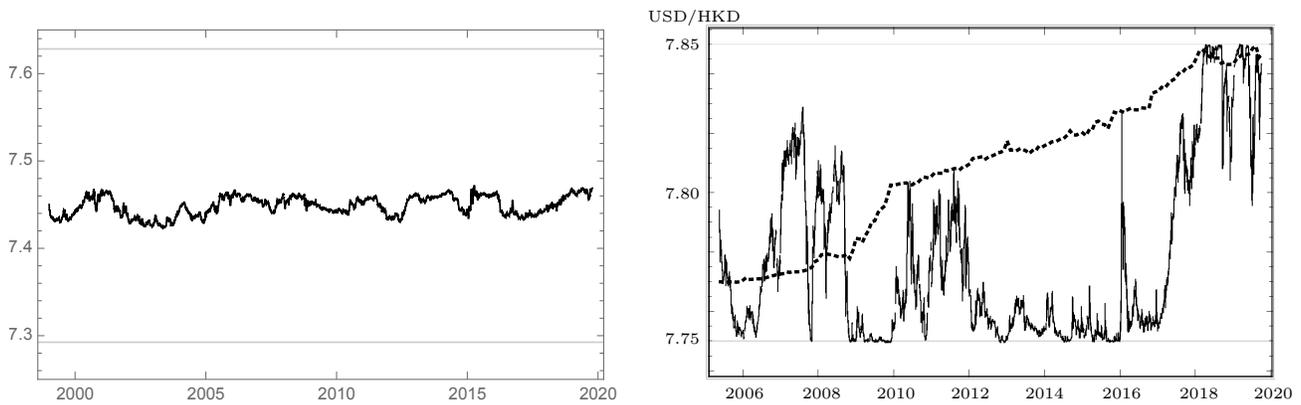
\begin{figure}[h]
\begin{center}
\begin{minipage}[b]{8.1cm}
\includegraphics[width=8.1cm]{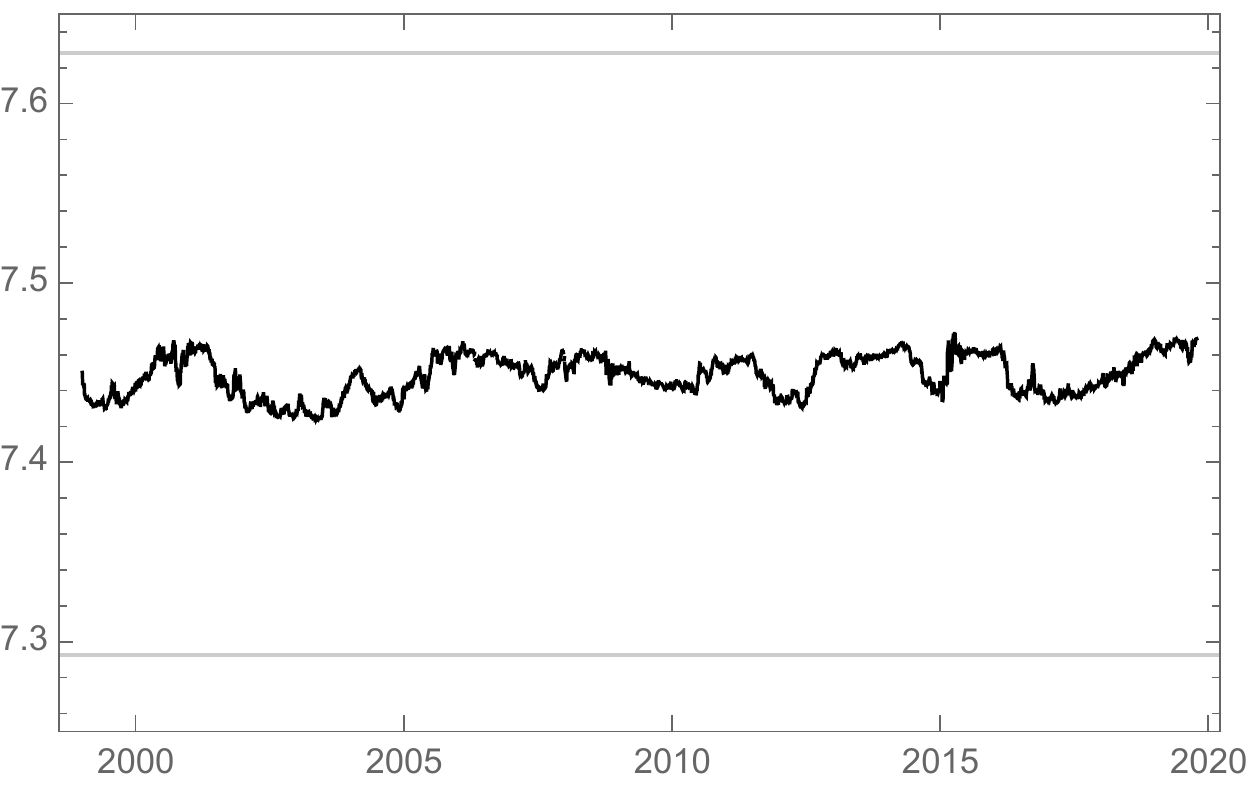}
\end{minipage}\qquad\quad \begin{minipage}[b]{8cm}
\begin{overpic}[width=7.6cm]{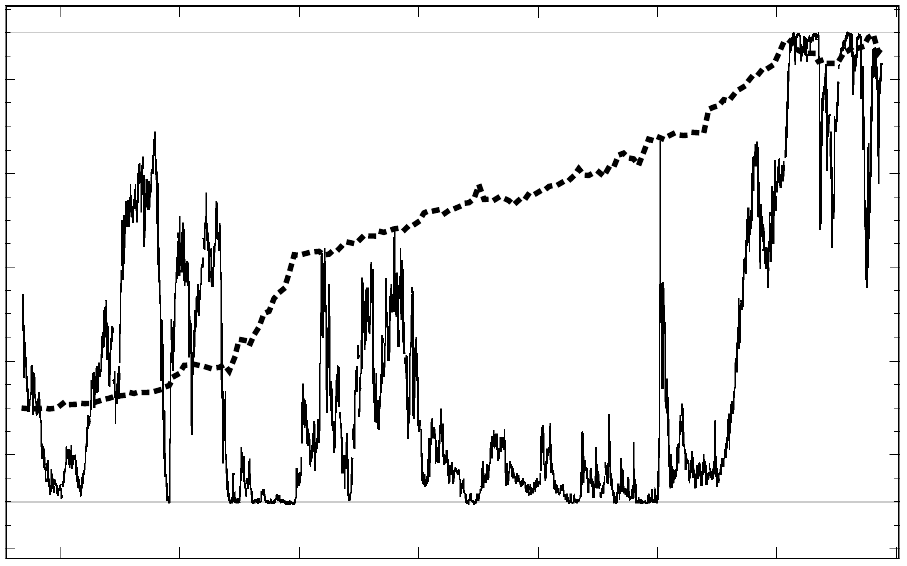}
\put(3.5,-3){\tiny 2006}
\put(17,-3){\tiny 2008}
\put(30,-3){\tiny 2010}
\put(43,-3){\tiny 2012}
\put(56,-3){\tiny 2014}
\put(69.5,-3){\tiny 2016}
\put(83,-3){\tiny 2018}
\put(96,-3){\tiny 2020}
\put(-7,6.5){\tiny 7.75}
\put(-7,32){\tiny 7.80}
\put(-7,58){\tiny 7.85}
\put(-10,63){\tiny USD/HKD}
\end{overpic}\\
 \vspace{-1.1mm}
\end{minipage}
\end{center}
\caption{ Left-hand panel: Since 1 January 1999, the target zone of the EUR/DKK currency pair is defined through the European Exchange Rate Mechanism II (ERM II) as a $\pm2.25\%$ band around a central rate of 7.46038.  Right-hand panel: The USD/HKD currency peg has been in place since 1983. 
Since 18 May 2005, its target zone  is defined through the band from 7.75 to 7.85.  
Note that the EUR/DKK exchange rate stays well away from the boundaries of the target zone, which are indicated by two horizontal gray lines, while the USD/HKD exchange rate spends a significant amount of time in the vicinity of both thresholds. The dotted line indicates Hong Kong's foreign exchange reserves, which increased from USD 122 billion in May 2005 to USD 439 billion in September 2019.   }\label{figure1}
\end{figure}

To keep the exchange rate from breaking the prescribed barriers, the central bank has  two main methods of intervention: 
\begin{enumerate}
\item It can change the interest rates for the domestic currency. 
\item  It can trade domestic versus foreign currency and  thereby generate permanent price impact. 
\end{enumerate}
 The instrument (a) is often fraught with political tensions and might have side effects on the economy as a whole. For instance, it was used only once\footnote{On December  18, 2014, and thus less than one month before the termination of the currency peg, the Swiss National Bank decreased its 3-months Libor target range from  $[0.0\%,0.25\%]$ to  $[-0.75\%,-0.25\%]$.  
Although this decision made news headlines by introducing a regime of strictly negative interest rates,   it had only a very minor effect on the exchange rate, which increased from 1.2010 on December 17 to  1.2039 on December 18.} during the entire duration of  the EUR/CHF currency peg.
 It may also not be very effective, as is illustrated by the fact that the dramatic announcement of increasing the British interest rates to the level of 15\% could not prevent the \lq\lq Black Wednesday" \cite{Guardian}. For this reason, we  will ignore the possible use of interest rate changes and focus instead on instrument (b), the generation of permanent price impact through trading foreign against domestic currency. 
 This instrument has the side effect of accumulating (long or short) inventory in the foreign currency. Historical examples of failed target zones illustrate that this inventory can become problematic, in particular when there is an adverse macroeconomic trend in the market. For instance, it was estimated in \cite{Guardian} that immediately before the British exit from the ERM, \lq\lq 40 per cent of Britain's foreign exchange reserves were spent in frenetic trading".  Another example is the more recent EUR/CHF target zone: In its article \emph{Why the Swiss unpegged the franc}, published on January 18, 2015, 
The Economist wrote:
\begin{quote} Thanks to this policy, by 2014 the SNB had amassed about \$480 billion-worth of foreign currency, a sum equal to about 70\% of Swiss GDP.
\end{quote}
Following the abandonment of the currency peg on January 15, 2015, the inventory of the Swiss National bank (SNB) lost about CHF 78 billion of its value, which is about 12\% of the Swiss GDP. These  figures are taken from Lleo and Ziemba \cite{LleoZiemba}, where one can also find the following quotation regarding the inventory of the Danish central bank accumulated through enforcing the  EUR/DKK target zone as shown in Figure \ref{figure1}:
\begin{quote}
In an effort to defend the peg, the central bank's currency reserves have soared between two-thirds and more than 100\% in recent months. They now amount to more than USD 110 billion, which is about one third of Denmark's 2014 GDP.
\end{quote}

To model permanent price impact in a pegged currency market, we will use a continuous-time market impact model of Almgren--Chriss-type, in which the exchange rate is a diffusion process controlled by the price impact generated by the central bank's strategy. The objective of the central bank is to enforce the target zone through  a strategy that minimizes the inventory in foreign currency. Within our model this objective is then formulated as  a stochastic control problem with random time horizon.  We transform the corresponding Hamilton--Jacobi--Bellman equation into a certain second-order ordinary differential equation with infinite boundary conditions. This boundary value problem is a special case of a class of elliptic partial  differential equations with infinite boundary values for which classical solutions were obtained by Lasry and Lions  \cite{Lasry-Lions}. Our main  result then states that, after reverting the transformation,  this classical solution provides indeed the solution to our optimal control problem. Finally, we provide numerical simulations of the optimally controlled exchange rate processes and the corresponding evolution of the central bank inventory. We illustrate that, depending on the model parameters, the controlled exchange rate process can be both of the type of the EUR/DKK and the USD/HKD exchange rates as shown in Figure \ref{figure1}.

The approach we exploit to defend the target zone is different from most of the literature. Many papers in economics assume that the target zone is defended with infinitesimal interventions at its edges, and no interventions take place when the exchange rate is strictly inside the band. This stream of literature typically considers equilibrium models  and focuses on the study of the exchange rate behavior; see, e.g., \cite{Krugman91}, the survey \cite{Portugal13}, and the references therein.  Some literature adopts the quasi-variational inequalities method to study the optimal impulse control (i.e., the optimal times and the optimal sizes of interventions) to defend a target zone of exchange rate (e.g., \cite{Jeanblanc-Picque1993} 
and \cite{Korn1997}). These papers commonly follow \cite{Krugman91} and describe the dynamics of the exchange rates by a ``fundamental" process plus a term proportional to the percentage change in the exchange rate. They study the optimal interventions on the fundamental process to retain the exchange rates within an exogenously given band at the minimum total transaction cost (including a fixed fee and proportional cost for each intervention). Some more recent literature \cite{Cadenillas-Zapatero1999, Cadenillas-Zapatero2000,Mundaca-Oksendal1998} does not take the target zone as exogenously given, and instead assumes the central bank aims at keeping the exchange rate as close as possible to a given target.  In \cite{Mundaca-Oksendal1998}, the exchange rate is controlled both by interventions in the foreign exchange market (via impulse controls) and by the determination of the domestic interest rates (via stochastic controls). The paper establishes a sufficient condition for a strategy to minimize the sum of  (i) the cost due to the deviations of the exchange rates and the domestic interest rates from their target levels and (ii) the cost due to transaction on the foreign exchange market.  By assuming a quadratic form for the deviation costs and a linear form for the transaction cost in the model of  \cite{Mundaca-Oksendal1998}, \cite{Cadenillas-Zapatero2000} derives a closed-form solution which explicitly indicates that it is optimal to maintain the exchange rate within a finite band.  Finally, in \cite{NeumanSchied18}, a Stackelberg equilibrium between a speculative investor and the central bank maintaining a one-sided target zone is established.

%\newpage
\section{Model setup and main result}\label{sec:mode}

Our goal is to set up a model in which a central bank intervenes in a currency market so as to keep the exchange rate within the domain $D=(\beta_-,\beta_+)$, where $\beta_-<\beta_+$ are two finite barriers. To this end, the exchange rate  will be modeled as a controlled  diffusion process $X^{x,u}=\{X^{x,u}(t)\}_{t\ge0}$, with the control $u$ describing the intervention strategy of the central bank and $x$ denoting the starting point. The intervention of the central bank will occur via the generation of permanent price impact through buying or selling the foreign currency. That is, buying one lot of the foreign currency will increase the exchange rate by a factor $\gamma>0$, whereas selling one lot decreases the exchange rate by the same amount. 
Specifically, let $B=\{B(t)\}_{t\ge0}$ be a standard Brownian motion on the probability space 
 $(\Omega, \mathcal{F},\mathbb{P})$ and let $ \{\mathcal{F}_{t}\}_{t\geq0}$ be right-continuous filtration generated by $B$. Following the continuous-time version \cite{Almgren} of the standard Almgren--Chriss model \cite{AlmgrenChriss2}, we assume that for a given constant $X(0)=x\in D$, the dynamics of $X^{x,u}$ are given by
 \begin{align}
X^{x,u}(t)=x+\sigma B(t)+\int_{0}^{t}b(X^{x,u}(r))\,dr+\gamma\int_{0}^{t}u(r)\,dr,\qquad t\ge0.
\label{eq-rate}
\end{align}
 Here,  $\sigma>0$ is the volatility, the drift $b(\cdot)$ describes a macro-economic trend of the exchange rate, and $u$ is a progressively measurable stochastic control process such that $\int_0^t(u(r))^2\,dr<\infty$ $\mathbb P$-a.s.~for all $t\ge0$ and such that the stochastic differential equation \eqref{eq-rate}  admits a unique  strong solution $X^{x,u}$ satisfying $X^{x,u}(t)\in D$ for all $t\ge0$. Any such control process $u$ will be called \emph{admissible}, and $\mathcal U$ will denote the set of all admissible control processes.

 The drift $b(\cdot)$ describes a  macro-economic tendency of investors to buy or sell the domestic currency. For instance, during the European sovereign debt crisis, such a trend resulted from the desire of many euro investors to seek  a safe haven in the Swiss franc. As a response, the Swiss National Bank introduced a lower currency peg of 1.20 on the EUR/CHF exchange rate, which held from 2011 to 2015. During that time, the Swiss National Bank had to defend the currency peg through purchasing large amounts of foreign currency. In the sequel, it will make sense to think of $b(\cdot)$ as being negative, although this assumption is not necessary from a mathematical perspective. On a technical level, we will assume that $b$ belongs to $C^1(D)$ and that both $b$ and its derivative $b'$ are bounded. A few additional assumptions will be formulated later. 
 
 We assume that the central bank has a random planning horizon $\tau$. This planning horizon will be reached when the currency peg is terminated, modified, or when the macroeconomic trend $b$ changes.   Such a change in the macroeconomic drift can be observed in the USD/HKD target zone in the right-hand panel of Figure \ref{figure1}: For more than a decade, the Monetary Authority of Hong Kong had only to defend the lower barrier of their target zone; in 2018, however, the emergence of the US-China trade war in 2018 changed that picture, and since then, the upper barrier has to be defended.

 Maintaining the target zone is associated with certain costs for the central bank. As discussed in the introduction, one of the main concerns is the accumulation of large inventory in the foreign currency, which presents a large risk in the event of a break-up of the target zone. The discounted total inventory accumulated until the random time $\tau$ equals
 $$\int_0^{\tau} u(t)e^{-\rho t}\,dt,
 $$
 where $\rho>0$ is the discount rate. In addition, the central bank incurs \lq\lq slippage", including transaction costs, instantaneous price impact effects etc., modeled by the cost functional 
  $$\eta\int_0^{\tau} (u(t))^2e^{-\rho t}\,dt,
 $$
where $\eta>0$ is the coefficient of the temporary price impact; see, e.g., \cite{Almgren} or Section 2.2 in \cite{Gatheral} for economic explanations of this cost term in the context of price impact models.  Thus, the central bank has the goal of minimizing the expected costs with random horizon 
\begin{align}
\label{eq-obj-bank-tau}J(x,u)=\mathbb{E}\bigg[\displaystyle\int_{0}^{\tau} \left(
u(t)+\eta(u(t))^{2}\right)  e^{-\rho t}\,dt\bigg]
\end{align}
over $u\in\mathcal U$. It is reasonable to assume that $\tau$ is an independent exponential time with parameter $\theta>0$.
Then the cost function can be rewritten as
\begin{equation}\label{eq-obj-bank}
\begin{split}
J(x,u)&=\mathbb{E}\bigg[\displaystyle\int_{0}^{\infty} \left(
u(t)+\eta(u(t))^{2}\right)  e^{-\rho t}\text{1}_{\{\tau>t\}}\,dt\bigg]\\
&=\mathbb{E}\bigg[\displaystyle\int_{0}^{\infty} \left(
u(t)+\eta(u(t))^{2}\right)  e^{-\lambda t}\,dt\bigg],
\end{split}
\end{equation}
where $\lambda:=\rho+\theta$. Finally, the value function is denoted by 
\begin{equation}\label{value function eq}
V(x)=\inf_{u\in\mathcal U}J(x,u),\qquad x\in D.
\end{equation}

\begin{remark}
In the financial interpretation of this cost functional, we use the above-mentioned idea that the macro-economic trend $b(\cdot)$ should be negative, so that it is primarily the lower barrier, $\beta_-$, that must be defended by the central bank. Therefore, it is  the long inventory in the foreign currency, which is problematic. If, as during the British exit from the ERM, the primary focus is on the upper barrier, $\beta_+$, then   the short inventory in the foreign currency is a concern, and we can simply replace the term $u(t)$ with $-u(t)$. Clearly, both scenarios are mathematically equivalent.
  \end{remark}

The usual heuristic arguments suggest that $V$ should satisfy the following Hamilton--Jacobi--Bellman (HJB) equation:
\begin{equation}\label{HJB rough eq}
\inf_{u\in \mathbb R}\left\{   V'(x)( b(x)+\gamma u)+\frac{1}{2}\sigma^2 V''(x)+u+\eta
u^{2}-\lambda V(x)\right\} =0,\qquad x\in D.
\end{equation}
By computing, at least informally, the infimum over $u$, we obtain the following reduced form of the HJB equation:
\begin{equation}\label{HJB reduced eq}
\frac{1}{2}\sigma^2 V''(x)+ V'(x) b(x)-\lambda
V(x)-\frac{(\gamma V'(x)+1)^2}{4\eta} =0,\qquad x\in D.
\end{equation}
Moreover, the minimizer  provides a candidate for the optimal Markovian strategy,
\begin{equation}\label{optimal strategy eq}
\widehat u(x)=-\frac{\gamma V'(x)+1}{2\eta }, \qquad x\in D.
\end{equation}
In addition to \eqref{HJB reduced eq}, we need to specify boundary values of $V$ at $\partial D=\{\beta_-,\beta_+\}$. Since the goal of the central bank is to keep the exchange rate within $D$, any value outside $ D$ should receive an infinite penalty, which amounts to the singular boundary condition
\begin{equation}\label{boundary condition}
V(x)\longrightarrow+\infty\quad\text{as $x\to\partial D$.}
\end{equation}
Since $D$ is a bounded domain, the boundary condition \eqref{boundary condition} will imply that the candidate strategy $\widehat u$ from \eqref{optimal strategy eq} will satisfy $\widehat u(x)\to+\infty$ as $x\downarrow\beta_-$ and $\widehat u(x)\to-\infty$ as $x\uparrow\beta_+$. Since $D$ is a bounded domain, the boundary condition~\eqref{boundary condition} implies that $V'(x)\to-\infty$ as $x\downarrow\beta_-$ and $ V'(x)\to+\infty$ as $x\uparrow\beta_+$, which yields that  the candidate strategy $\widehat u$ from~\eqref{optimal strategy eq} satisfies $\widehat u(x)\to+\infty$ as $x\downarrow\beta_-$ and $\widehat u(x)\to-\infty$ as $x\uparrow\beta_+$. This property of $\widehat u$  is natural, because only unbounded controls $u$ will be able to keep the controlled diffusion $X^{x,u}(t)$ inside the target zone $ D$ for all $t\ge0$. Now we can formulate our main result.

\begin{theorem}\label{th-1} Suppose that 
 $b(\cdot)$ is bounded,  twice continuously differentiable, and Lipschitz continuous. Then:
 \begin{enumerate}
 \item\label{thm part (a)} The singular boundary value problem \eqref{HJB reduced eq}, \eqref{boundary condition} admits a unique classical solution $V$. Moreover, at the boundary, the solution blows up like the logarithmic distance to $\partial D$, i.e.,
 \begin{equation}\label{V log behavior eq}
\lim_{\eps\downarrow0}\frac{V(\beta_\pm\mp\eps)}{-\log\eps}=\frac{2\sigma^2\eta}{\gamma^2}.
 \end{equation}
 \item The function $V$ from part \ref{thm part (a)} is equal to the value function \eqref{value function eq}. Moreover,   the optimal control is Markovian and given in feedback form by 
 \begin{equation}\label{optimal control}
 u^*(t)=-\frac{\gamma V'(X^{x,u^*}(t))+1}{2\eta }.
 \end{equation}
   \end{enumerate}
 \end{theorem}

Part \ref{thm part (a)} of Theorem \ref{th-1}  is based  on a transformation of the singular boundary value problem \eqref{HJB reduced eq}, \eqref{boundary condition} into a related equation for which existence and uniqueness of classical solutions were studied by Lasry and Lions \cite{Lasry-Lions}. This transformation is summarized in the following proposition.

\begin{proposition}\label{trafo prop}Under the same assumptions as Theorem~\ref{th-1}, fix some $x_0\in D$. A function $V$ solves \eqref{HJB reduced eq} if and only if 
\begin{equation}\label{VW relation}
W(x):=\frac\gamma{2\eta\sigma^2}\Big(\gamma V(x)+x-\frac{2\eta}\gamma\int_{x_0}^xb(y)\,dy\Big)
\end{equation}
solves 
\begin{equation}\label{LL eq}
-W''(x)+\big(W'(x)\big)^2+\frac{2\lambda}{\sigma^2}W(x)=f(x),\qquad x\in D,
\end{equation}
for
\begin{equation}\label{f eq}
f(x)=-\frac{\gamma}{\sigma^4\eta}b(x)+\frac{1}{\sigma^2}b'(x)+\frac{1}{\sigma^4}(b(x))^2+\frac{\gamma \lambda}{\sigma^4\eta}x-\frac{2\lambda}{\sigma^4}\int_{ x_0}^xb(y)\,dy.
\end{equation}
\end{proposition}

\section{{Numerical experiments}}

A first initial guess for solving the singular boundary value problem \eqref{HJB reduced eq}, \eqref{boundary condition} is to replace the singular boundary condition \eqref{boundary condition}
 with large but finite boundary values. This is indeed possible, as  it was shown in \cite{Lasry-Lions} that the boundary value problem defined through \eqref{LL eq} and $W(\beta_\pm)=R$ admits at least a weak solution for each $R>0$. By Proposition \ref{trafo prop}, this solution can then be transformed into a solution of our original differential equation \eqref{HJB reduced eq}. We found, however, that this method is numerically unstable for large $R$. Moreover, it  does not accurately reflect the blow-up of solutions at the boundary, which is important for the requirement that the controlled diffusion stays within the domain $D=(\beta_-,\beta_+)$. 
  To overcome these disadvantages, we used the following method, which provides a numerical solution that retains the boundary behavior of $V$. 
  
   Let $d(\cdot)$ be any function in $C^2(\overline D)$ that satisfies $d(x)>0$ for all $x\in D$, $d(\beta_\pm)=0$, $d'(\beta_-)=+1$, and $d'(\beta_+)=-1$. For instance, one can take 
   \begin{equation}\label{quadratic d}
   d(x)=\frac1{\beta_+-\beta_-}\big(-\beta_-\beta_++(\beta_-+\beta_+)x-x^2\big)
   \end{equation}
   or
$$d(x)=\frac{\beta_+-\beta_-}\pi\sin\Big(\frac{\pi(x-\beta_-)}{\beta_+-\beta_-}\Big).
$$
Then it follows from  \eqref{V log behavior eq} that, when approaching the boundary $\partial D$, the value function behaves like a constant times $({-\log d(x)})$. Hence, if we define 
$$U(x):=\frac{V(x)}{-\log d(x)},
$$
then the equation \eqref{HJB reduced eq} is equivalent to the following equation for $U$,
\begin{equation}\label{U ODE}
\begin{split}
&4 \eta  b    \left(U  dd' +U'd^2  \log d   \right)+\left(\gamma 
   U  d' +  \gamma U' d \log d  -d\right)^2\\
&\quad   +2 \eta  \sigma^2
   \left(  2 dd'  U' + U''d^2 \log d  +U  \left(d 
   d'' -(d')^2\right)\right)-4 \eta  \lambda   U  d ^2\log d =0,
\end{split}
\end{equation}
and the asymptotic boundary behavior \eqref{V log behavior eq} translates into the following two-point boundary condition,
\begin{equation}\label{U boundary cond}
U(\beta_\pm)=\frac{2\sigma^2\eta}{\gamma^2}.
\end{equation}
After solving the boundary value problem \eqref{U ODE} and \eqref{U boundary cond} numerically, a numerical solution for the boundary value problem \eqref{HJB reduced eq}, \eqref{boundary condition}  is then obtained by setting 
\begin{equation}\label{log trafo}
V(x):=-U(x)\log d(x)
\end{equation}
 and the optimal Markovian control is given by
$$\widehat u(x)=\frac{\gamma }{2\eta}\Big(\frac{U(x)d'(x)}{d(x)}+ U'(x)\log d(x)-\frac1\gamma\Big)
$$
This method clearly retains the blow-up rate \eqref{V log behavior eq} of $V$.

For solving the auxiliary boundary value problem  \eqref{U ODE} and \eqref{U boundary cond} numerically, we chose the function $d$ as in \eqref{quadratic d}. To avoid the singularity of the coefficients of the differential equation \eqref{U ODE} at the boundary, however, we replaced $d(x)$ with $d_\eps(x):=\eps+(1-\eps)d(x)$ in \eqref{U ODE}, but not in \eqref{log trafo}; in our simulations we took $\eps=0.00176$. Then we  used the Mathematica routine {\tt NDSolve} to numerically solve the corresponding boundary value problem. The corresponding exchange rate and inventory processes can then be simulated by means of a standard Euler scheme. When doing so, it is important to choose a small step size, so as to avoid that increments can jump over the barriers $\beta_-$ and $\beta_+$. In our simulations, we took a step size of $1.6\times 10^{-6}$.

The results of our simulations are shown in Figures \ref{figure2}
 and \ref{figure3}. These two figures illustrate that, depending on the model parameters, our central bank strategies can exhibit the features of both the EUR/DKK and the USD/HKK exchange rates as shown in Figure \ref{figure1}. That is, in Figure \ref{figure2}, the controlled diffusion stays well away from the boundary $\partial D=\{\beta_-,\beta_+\}$, whereas in Figure \ref{figure3}, the process spends a significant amount in the vicinity of the lower bound $\beta_-$. Moreover,  in Figure \ref{figure3}, the inventory process appears to be close to  the local time of a diffusion reflecting at $\beta_-$, with essentially linear decay during excursions from the reflecting barrier. This can justify the common use of reflecting diffusions in target zone models; see, e.g.,  \cite{Ball-Roma98,NeumanSchied} and the references therein. We mention finally that by changing the model parameters  it is also possible to obtain less regularly shaped solutions $V$ and in particular non-monotone Markovian controls $\widehat u$. Our choice of the plots in Figures \ref{figure2}
 and \ref{figure3} was motivated by their similarity to the historical EUR/DKK and USD/HKD in Figure \ref{figure1}.

\begin{figure}[h]
\begin{center}
\begin{minipage}[b]{6cm}
\includegraphics[width=6cm]{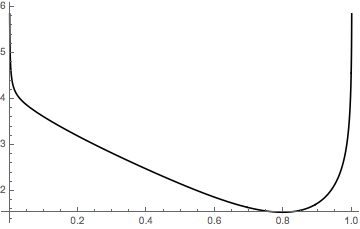}
\end{minipage}\qquad \begin{minipage}[b]{8cm}
\includegraphics[width=6cm]{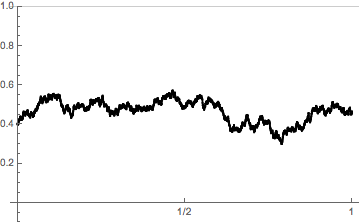}
\end{minipage}\\
\begin{minipage}[b]{6cm}
\includegraphics[width=6cm]{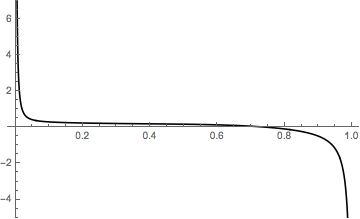}
\end{minipage}\qquad \begin{minipage}[b]{8cm}
\includegraphics[width=6cm]{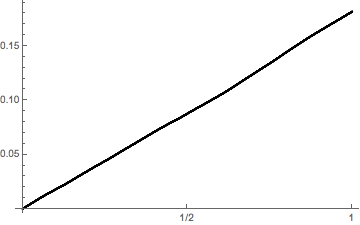}
\end{minipage}
\end{center}
\caption{ Solution $V$ (top left), optimal Markovian control $\widehat u$ (bottom left), as well as a sample path of the optimally controlled exchange rate, $X^{0.4,u^*}(t)$ (top right),  and the corresponding inventory $\int_0^tu^*(s)\,ds$ (bottom right), for $0\le t\le1$. The parameters are $\beta_-=0$, $\beta_+=1$, $\sigma=0.25$, $\eta=6$, $\gamma=1$, $\lambda=0.5$, and $b(x)=-(1-x^2)/2$. }\label{figure2}
\end{figure}

\begin{figure}[h!]
\begin{center}
\begin{minipage}[b]{6cm}
\includegraphics[width=6cm]{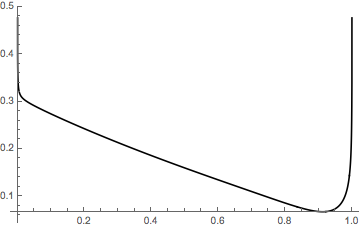}
\end{minipage}\qquad \begin{minipage}[b]{8cm}
\includegraphics[width=6cm]{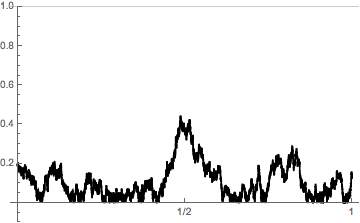}
\end{minipage}\\
\begin{minipage}[b]{6cm}
\includegraphics[width=6cm]{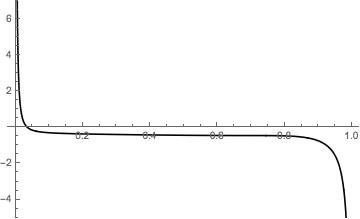}
\end{minipage}\qquad \begin{minipage}[b]{8cm}
\includegraphics[width=6cm]{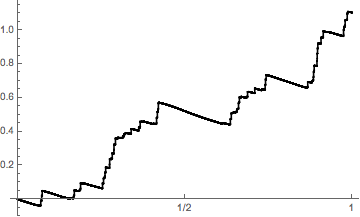}
\end{minipage}
\end{center}
\caption{ Solution $V$ (top left), optimal Markovian control $\widehat u$ (bottom left), as well as a sample path of the optimally controlled exchange rate, $X^{0.2,u^*}(t)$ (top right),  and the corresponding inventory $\int_0^tu^*(s)\,ds$ (bottom right), for $0\le t\le1$.  The parameters are $\beta_-=0$, $\beta_+=1$, $\sigma=0.4$, $\eta=0.6$, $\gamma=2$, $\lambda=1$, and $b(x)=-(1-x^6)$. }\label{figure3}
\end{figure}

\section{Proofs}

\begin{proof}[Proof of Proposition \ref{trafo prop}] By completing the squares, we see that   \eqref{HJB reduced eq} is equivalent to
\begin{equation}\label{HJB completed squares eq}
\frac{1}{2}\sigma^2 V''(x)-\frac1{4\eta}\Big(\gamma V'(x) +1-\frac{2\eta}\gamma  b(x)\Big)^2 -\frac1\gamma b(x) +\frac\eta{\gamma^2}(b(x))^2 -\lambda
V(x)=0.
\end{equation}
From the given relation \eqref{VW relation} between $V$ and $W$, we get
$$
V(x)=\frac1\gamma\Big(\frac{2\eta\sigma^2}\gamma W(x)-x+\frac{2\eta}\gamma\int_{x_0}^xb(y)\,dy\Big).
$$ 
Taking the first two derivatives of the right-hand side of the preceding equation and plugging them into \eqref{HJB completed squares eq} yields that $V$ satisfies \eqref{HJB completed squares eq}  if and only if $W$ satisfies
$$\frac{\sigma^4\eta}{\gamma^2} W''(x)-\frac{\sigma^4\eta}{\gamma^2}\big(W'(x)\big)^2-\frac{2\sigma^2\eta \lambda}{\gamma^2}W(x)=\frac1\gamma b(x)-\frac{\eta \sigma^2}{\gamma^2}b'(x)-\frac{\eta}{\gamma^2}(b(x))^2-\frac{\lambda }{\gamma} x+\frac{2\lambda\eta}{\gamma^2}\int_{x_0}^xb(y)\,dy.
$$
Dividing by $-{\sigma^4\eta}/\gamma^2$ shows that the preceding equation is the same as \eqref{LL eq}.\end{proof}

\begin{proof}[Proof of Theorem \ref{th-1}] Our assumptions imply that the function $f$ defined in \eqref{f eq} is bounded and continuously differentiable. Under this condition, 
the existence of a unique classical solution $W$ to \eqref{LL eq} with boundary conditions
\begin{equation}
\lim_{\eps\downarrow0}\frac{W(\beta_\pm\mp\eps)}{\log\eps}=-1
 \end{equation}
was shown in 
\cite[Theorem I.1]{Lasry-Lions}. Part \ref{thm part (a)} of the assertion hence follows from Proposition \ref{trafo prop}.

The proof of part (b) is based on the ideas from \cite[Section VII]{Lasry-Lions}. Note, however, that the quadratic case, which is the one that is relevant for us here, is not included in  the statement of the results from \cite{Lasry-Lions}.

Take a fixed starting point $x\in D$. Let $V$ denote the solution of the singular boundary value problem \eqref{HJB reduced eq}, \eqref{boundary condition}  as provided by part \ref{thm part (a)} and 
$\widehat u(x)=-\frac1{2\eta }(\gamma V'(x)+1)$ so that $u^*(t)=\widehat u(X^{x,u^*}(t))$ by \eqref{optimal control}.  We show first that $u^*$  is indeed an admissible control, i.e., $u^*\in\mathcal U$. For  $\delta>0$, let  $D_{\delta}=(\beta_-+\delta,\beta_+-\delta)$ and  $\tau_{\delta}$ be the first exit
time of $X^{x,u^*}$ from $D_\delta$, where $\delta$ should be small enough so that $x\in D_\delta$. For the ease of notation, we will also write $X$ for  $X^{x,u^*}$. Note that $V$ and its derivatives are bounded on $D_\delta$. Hence, applying It\^o's formula to
$e^{-\lambda t}V(X(t))$ and taking expectations gives
\begin{equation}
\begin{split}
V(x)&=\mathbb{E}\bigg[\int_{0}^{\tau_{\delta}}e^{-\lambda t}\Big( \lambda V(X(t))- V'(X(t)) [b(X(t))+\gamma\widehat u(X(t))]-\frac{\sigma^2}{2}V''(X(t)) \Big)\,dt\bigg]\\
&\qquad+\mathbb{E}[e^{-\lambda \tau_{\delta}}V(X(\tau_{\delta})) ].\\
\end{split}
\end{equation}
Using the HJB equation \eqref{HJB rough eq} and the fact that $\widehat u$ attains the infimum, we obtain
\begin{align} 
V(x)=\mathbb{E}\bigg[\int_{0}^{\tau_{\delta}}e^{-\lambda t} \left(\widehat u (X(t))+\eta (\widehat u (X(t)))^2 \right)\,dt\bigg]
+\mathbb{E}[e^{-\lambda \tau_{\delta}}V(X(\tau_{\delta})) ].\label{eq-value-u}
\end{align}
The first term on the right-hand side can be estimated as follows,
\begin{align*}
\mathbb{E}\bigg[\int_{0}^{\tau_{\delta}}e^{-\lambda t} \left(\widehat u (X(t))+\eta (\widehat u (X(t)))^2 \right)\,dt\bigg]
&=\mathbb{E}\bigg[\int_{0}^{\tau_{\delta}}e^{-\lambda t}\eta \left(\Big((\widehat u (X(t))+\frac{1}{2\eta}\Big)^2 -\frac{1}{4 \eta^2}\right)\,dt\bigg]\\&\geq -\frac{1}{4 \eta} \mathbb{E}\bigg[\int_{0}^{\tau_{\delta}}e^{-\lambda t}dt\bigg]\geq  -\frac{1}{4 \eta \lambda}.
\end{align*}
The second term on the right-hand side of \eqref{eq-value-u} can estimated as follows,
$$\mathbb{E}[e^{-\lambda \tau_{\delta}}V(X(\tau_{\delta})) ]\ge \big(V(\beta_-+\delta)\wedge V(\beta^+-\delta)\big)\mathbb E[e^{\-\lambda\tau_{\delta}}].
$$
Therefore,
$$V(x)\ge  -\frac1{4\eta\lambda}+ \big(V(\beta_-+\delta)\wedge V(\beta^+-\delta)\big)\mathbb E[e^{\-\lambda\tau_{\delta}}].
$$
Since $V(x)$ is finite in $D_\delta$ and $V(\beta_-+\delta)\wedge V(\beta^+-\delta)\to+\infty$ as $\delta\downarrow0$, we must have  
$$0=\lim_{\delta\downarrow0}\mathbb{E}[e^{-\lambda \tau_{\delta}}]=\mathbb{E}[e^{-\lambda \tau_0}] ,$$
where $\tau_0:=\lim_{\delta\downarrow0}\tau_\delta$ is equal to the exit time of $X$ from $D$. 
Therefore, we must have $\tau_0=\infty$ $\mathbb P$-a.s., and so $u^*$  is admissible.

%We will complete the proof by showing the optimality of $u^\ast(\cdot)$ and the identity between $\tilde{V}$ and $V$. 

Now we show that 
\begin{equation}\label{V ge J}
V(x)\ge J(x,u^*)\ge\inf_{u\in\mathcal U}J(x,u).
\end{equation}
The first inequality will yield in particular that $J(x,u^*)$ is finite. For sufficiently small $\delta$, we have  $V(x)\geq 0$ on $D\setminus D_{\delta}$, and thus, by (\ref{eq-value-u}),   
$$
V(x)\geq \mathbb{E}\bigg[\int_{0}^{\tau_{\delta}}e^{-\lambda t} \left(u^*(t)+\eta (u^*(t))^2 \right)\,dt\bigg].
$$
Since $\tau_{\delta} \rightarrow +\infty$ as $\delta\downarrow 0$, we obtain $
V(x)\geq J(x,u^*)$ and thus \eqref{V ge J}. 

In the remainder of the proof, we show
\begin{equation}\label{V le inf J}
V(x)\leq J(x,u), \qquad\text{for all $u\in \mathcal{U}$},
\end{equation}which  in turn gives the optimality of $u^{\ast}$. To this end, we use the fact, established in Step 2 of the proof of \cite[Theorem II.1]{Lasry-Lions}, that for each $n\in\mathbb N$, the following two-point boundary value problem admits a solution $W_n$ in the Sobolev space $W^{2,2}(D)$,
\begin{equation}\label{trun-W}
\begin{cases}\displaystyle
-W_{n}''(x)+\big(W_{n}'(x)\big)^2+\frac{2\lambda}{\sigma^2}W_{n}(x)&\!\!=f(x)\quad\text{for $x\in D$},\\
\hfill W_{n}(\beta_{\pm}) &\!\!=n.\end{cases}
\end{equation}
In our one-dimensional situation, functions in $W^{2,2}(D)$ are almost everywhere differentiable with a square-integrable derivative and hence extend to continuous functions on $\overline D=[\beta_-,\beta_+]$. The same is true of their derivatives, so that $W^{2,2}(D)\subset C^1(\overline D)$. In \cite{Lasry-Lions}, it is shown moreover that the functions $W_n$ increase pointwise to the continuous function $W$, so that $W_n\to W$ uniformly on compact subsets of $D$.

 When defining
$$
V_n(x):=\frac{2\eta\sigma^2}{\gamma^2}W_n(x)-\frac x\gamma+\frac{2\eta}{\gamma^2}\int_{x_0}^xb(y)\,dy,
$$
Proposition \ref{trafo prop} thus yields a sequence of solutions to \eqref{HJB reduced eq} that belong to $W^{2,2}(D)\cap C^1(\overline D)$ and that converge to $V$ uniformly on compact subsets of $D$. Moreover, It\^o's formula also applies to such functions (see, e.g., \cite[Remark 3.2]{FoellmerProtterShiryaev}). 

We now take $u\in\mathcal U$ such that $J(x,u)<\infty$, let $X(t):= X^{x,u}(t)$, apply It\^o's formula to $e^{-\lambda t}V_n(X(t))$, and take expectations to get
\begin{align*}
V_n(x)&=\mathbb E\bigg[e^{-\lambda\tau_\delta}V_n(X(\tau_\delta))-\int_0^{\tau_\delta}e^{-\lambda t}\Big(\big(b(X(t))+\gamma u(t)\big)V_n'(X(t))-\lambda V_n(X(t))+\frac{\sigma^2}2V_n''(X(t))\Big)\,dt\bigg]\\
&\le \mathbb E\big[e^{-\lambda\tau_\delta}V_n(X(\tau_\delta))\big]+\mathbb E\bigg[\int_0^{\tau_\delta}e^{-\lambda t}\Big(u(t)+\eta\big(u(t)\big)^2\Big)\,dt\bigg],
\end{align*}
where we have used \eqref{HJB rough eq} in the second step. Since $\mathbb P$-a.s.~$X(t)\in D$ for all $t\ge0$, we must have $\tau_\delta\to\infty$ as $\delta\downarrow0$, and so 
$$
\big|\mathbb{E}[e^{-\lambda \tau_{\delta}}{V}_n(X(\tau_{\delta})) ]\big|\leq \sup_{y\in \overline{D}} |{V}_n(y)| \cdot \mathbb{E}[e^{-\lambda \tau_{\delta}}] \longrightarrow 0\quad\text{\ as\ } \delta \downarrow 0.
$$
Consequently, 
\begin{equation*}
V_n(x) \leq \mathbb E\bigg[\int_0^{\infty}e^{-\lambda t}\Big(u(t)+\eta\big(u(t)\big)^2\Big)\,dt\bigg]=J(x,u).
\end{equation*}
Sending $n\uparrow \infty$, we obtain \eqref{V le inf J}.  
\end{proof}

\noindent \textbf{Acknowledgement}. E.N.~would like to thank the CFM -- Imperial
Institute of Quantitative Finance which supported this research. A.S.~and X.X.~gratefully acknowledge financial support from the
 Natural Sciences and Engineering Research Council of Canada (NSERC) through grant RGPIN-2017-04054. C.W.~and X.X.~would like to thank the financial support from NSERC through grant RGPIN-2016-04001.
X.X.~is
grateful to  Peng Liu, Peng Luo and Wei Wei for  helpful discussions. All authors thank Abel Cadenillas for kindly pointing out   the references  \cite{Cadenillas-Zapatero1999, Cadenillas-Zapatero2000,Jeanblanc-Picque1993,Korn1997, Mundaca-Oksendal1998}.

\bibliographystyle{abbrv}
\bibliography{MarketImpact}

\end{document}